\documentclass[amsmath,amssymb,superscriptaddress,10pt]{revtex4}
\usepackage{amsthm}
\usepackage{enumerate}

 \newtheorem{theorem}{Theorem}[section]
 \newtheorem{lemma}[theorem]{Lemma}

\newcommand*{\bbC}{\mathbb{C}}
\newcommand*{\complex}{\bbC}

\newcommand*{\cC}{\mathcal{C}}
\newcommand*{\cD}{\mathcal{D}}
\newcommand*{\cE}{\mathcal{E}}

\newcommand*{\cH}{\mathcal{H}}

\newcommand*{\cM}{\mathcal{M}}

\newcommand*{\cP}{\mathcal{P}}
\newcommand*{\cS}{\mathcal{S}}

\newcommand*{\bracket}{\mathcal{M}\mbox{-map}}

\newcommand*{\tr}{\mathsf{tr}}
\newcommand*{\ket}[1]{|#1\rangle}
\newcommand*{\bra}[1]{\langle #1|}
\newcommand*{\proj}[1]{\ket{#1}\bra{#1}}

\newcommand{\braket}[2]{\langle #1|#2\rangle}       

\begin{document}

\title{Cumulants and the moment algebra: tools for analysing
weak measurements}

\author{Johan \surname{{\AA}berg}}
\email[]{jaaberg@phys.ethz.ch} 
\affiliation{Institute for Theoretical Physics, ETH Zurich,
             8093 Zurich, Switzerland}
\affiliation{Communication Technology Laboratory, ETH Zurich, 8092 Zurich, Switzerland}

\author{Graeme \surname{Mitchison}}
\email[]{g.j.mitchison@damtp.cam.ac.uk} \affiliation{Centre for
Quantum Computation, DAMTP,
             University of Cambridge, 
             Cambridge CB3 0WA, UK}

\begin{abstract}
Recently it has been shown that cumulants significantly simplify the
analysis of multipartite weak measurements. Here we consider
the mathematical structure that underlies this, and find that
it can be formulated in terms of what we call the moment algebra.
Apart from resulting in simpler proofs, the flexibility of this
structure allows generalizations of the original results to a number
of weak measurement scenarios, including one where the weakly
interacting pointers reach thermal equilibrium with the probed system.
\end{abstract}

\maketitle




\section{Introduction}

For many readers, the word ``cumulant'', if it means anything,
probably evokes a slight feeling of discomfort: a recollection,
perhaps, of a baffling definition and a paper left half-read. Yet, as
we hope to show here, cumulants should have pleasurable
associations. They arise as part of an algebraic structure, the moment
algebra, that can be defined very simply yet has striking
properties. It has the familiar operations of complex analysis -- a
multiplication and inverse, functions like $\log$ and $\exp$, a
derivative operation, power series expansions, etc. -- but all
transposed into a very different setting, with functions defined on a
lattice of finite subsets instead of the continuum of a complex
space, and with curious-looking new definitions for the
functions.  In the moment algebra, the cumulant is just the log
function, though many textbooks do an impeccable job of concealing
this fact.

Cumulants have a long history, with roots in statistics. They
were probably first considered by Thorwald N. Thiele
\cite{Thiele1,Thiele2} as ``half-invariants''; they then went through
a protean sequence of name changes \cite{Fisher1,Craig,Sylvester}
until the current name \cite{Fisher2} finally stuck.  Since they are
tools of statistics and probability theory \cite{KendallStuart77}, it
is perhaps unsurprising that they have been applied in statistical
mechanics \cite{Kubo, KahnUhlenbeck, Sylvester, Royer}, notably for
the calculation of virial coefficients and perturbation expansions of
the free energy, as well as in solid state physics, quantum chemistry, and
 quantum field theory (see e.g., \cite{Fulde, Kladko, Archontis,
  Wenbook}). Other studies have focused more directly on the cumulants
as, in some sense, genuine multipartite correlation measures, both in
a classical setting \cite{Carruthers}, and a quantum setting
\cite{Zhou}.

The second component of this paper is weak measurement
\cite{AharonovRohrlich05, AAV88}.  This is a way of obtaining
information about a system while perturbing it only a little,
by coupling the system weakly to a pointer. The imprecision of the
measurement outcome is compensated for by running many repeats of the
protocol, each time with a freshly prepared system. The measurement
results obtained this way are often surprising \cite{YH3}, and give
insight into the underlying physics, as in the analysis \cite{Aha} of
Hardy's paradox \cite{Hardy}. Other examples include phenomena in
fiber optics \cite{Brunner} and photonic crystals \cite{Solli}.

Here we consider multipartite weak measurements, by which we mean
measurements involving more than one pointer
\cite{ReschSteinberg,Resch, Lundeen,MJP07}. The moments of pointer
observables (i.e. the expectations of products of those observables)
turn out to depend in an extremely complicated way on  weak
values. Despite this, it has recently been shown that the cumulants of
pointer moments are very simply related to cumulants of weak values
\cite{GraemeCumulants}. This suggests that the role cumulants play in
simplifying perturbation expansions in statistical mechanics may have
an analogue in weak measurement. Note that, given the weak value
cumulants, one may if one wishes obtain the weak values themselves by
the inverse operation to the cumulant (the exponential in the moment
algebra); this gives an operational procedure for computing weak
values.

With the help of the moment algebra, we uncover some of the
mathematical structure that underlies the favourable interplay between
cumulants and weak measurements.  We also show that these results can
be generalized by relaxing some of the assumptions behind weak
measurement.  For instance, we consider a situation where weak
measurements are performed over an extended period during which the
system undergoes continuous evolution.  This is related, by what can
be broadly described as an imaginary time transformation, to another
scenario where the pointers and the system correlate via
thermalization; we call this ``thermal weak measurement''.  Thus in
general we get a broader view of weak measurement, together with some
tools for making its analysis more tractable.




\section{\label{convalgebra}The moment algebra and cumulants}

Here we introduce what we call the moment algebra. This can be
regarded as a natural setting for discussions about cumulants, and
provides a handy formalism for the proofs in the rest of this paper.

Let $\Omega_n$ be the set of integers $1, 2, \dots, n$, for some
$n>0$.  Let $\cM_n$ denote the set of functions that assign a
complex number $f(a) \in \complex$ to every subset $a$ of $\Omega_n$,
including the empty set $\emptyset$. We will refer to such an $f$ as
an $\bracket$.

If $f$ and $g$ are two $\bracket$s in $\cM_n$, we can define their
product $fg$ simply as the new $\bracket$ $[fg](a)=f(a)g(a)$, for any
subset $a$ of $\Omega_n$. However, there is another product, the {\em
  convolution product} $f*g$, which has particularly interesting
properties. If $a=\{a_1, \ldots, a_k\}$, let $\partial_a f$ denote the
formal derivative $(\partial/\partial \xi_{a_1} \ldots
\partial/\partial \xi_{a_k}) f$, where the $\xi_i$ are notional
variables that we never deal with explicitly. We now define
\begin{align}
\label{product} (f*g)(a)=\partial_a(fg),
\end{align}
where the right hand side is interpreted as follows: In the standard
expression for the derivative of the product $fg$, we make the
replacements $\partial_{b}f \to f(b)$ and $\partial_{b}g \to g(b)$ for
any subset $b \subset a$; i.e. we replace the derivative
$\partial_{b}f$ by the value $f(b)$ of the $\bracket$ f on $b$,
and similarly with $\partial_{b}g$. We also replace plain $f$ by
$f(\emptyset)$. As an example, suppose $a=\{1,2\}$. Then we have
\begin{align}
\label{diff} \partial_{1,2}(fg)=(\partial_{1,2}f)g+(\partial_1f)(\partial_2g)+(\partial_2f)(\partial_1g)+f(\partial_{1,2}g),
\end{align}
and after replacing the derivatives we obtain the convolution product evaluated at $a = \{1,2\}$,
\begin{align}
\label{product2} (f*g)(1,2)=f(1,2)g(\emptyset)+f(1)g(2)+f(2)g(1)+f(\emptyset)g(1,2).
\end{align}
In general, (\ref{product}) gives the explicit rule
\begin{align}
\label{explicit-product} (f*g)(a)&=\sum_{a_1 \cup a_2=a}f(a_1)g(a_2)
\end{align} 
for any subset $a\neq \emptyset$, where the sum runs over all ordered
bipartitions $(a_{1},a_{2})$ of $a$, including $(\emptyset,a)$ and
$(a,\emptyset)$ (treated as distinct). In the case $a =\emptyset$ we
find $(f*g)(\emptyset) = f(\emptyset)g(\emptyset)$.  The vector
space of $\bracket$s, $\cM_n$, together with the convolution product
becomes an algebra, which we also denote by $\cM_n$ and call the {\em
moment algebra}.

The terms ``moment'' and ``convolution'' are inspired by the following
construction.  Let $\phi(x_{1},\ldots x_{n})$ be a complex valued
integrable function of $n$ complex variables. Define an $\bracket$ by
\begin{equation}
\label{phi-map} f_{\phi}(a) = \int x_{a_{k}} \cdots
x_{a_{1}}\phi(x_1, \ldots, x_n)dx_1 \ldots dx_n,
\end{equation}
for each subset $a = \{a_{k},\ldots, a_{1}\}$. Hence, the $\bracket$
$f_{\phi}$ assigns to each subset $a$ the moment of the function
$\phi$ corresponding to $a$.  If we now choose another function
$\psi$, then $f_{\phi}*f_{\psi}$ gives moments of the usual
convolution $\phi*\psi(y_1, \ldots, y_n)=\int \phi(x_1, \ldots,
x_n)\psi(y_1-x_1, \ldots, y_n-x_n)dx_1 \ldots dx_n$. We can write this
succinctly as $f_{\phi}*f_{\psi}=f_{\phi*\psi}$, where the subscript
`$\phi*\psi$' is the usual convolution product.  In fact, any
$\bracket$ can be represented by an $f_{\phi}$ in this way, for some
(non-unique) $\phi$, but we do not make use of this, and proceed
entirely within our abstract algebra framework.

There is an identity, $1^*$ in $\cM_n$, defined by
\begin{align}\label{identity}
1^*(\emptyset)=1, \ \ 1^*(a)=0 \ \mbox{ for }\  a \ne \emptyset.
\end{align}
We have
\begin{align}
f*1^*=1^**f=f.
\end{align}
Given a sufficiently differentiable mapping
$F:\mathbb{C}\rightarrow\mathbb{C}$ we can define a mapping
$F^{*}:\cM_n\rightarrow\cM_n$ via
\begin{equation}
\label{Fstar} (F^{*}f)(a) = \partial_{a}F(f),
\end{equation}
where we assume that the formal derivative operates according to the
standard chain rule, giving $\partial_a F(f)$ as a function of the
formal derivatives $\partial_b f$ of $f$. Note that $F^{*}f(\emptyset) = F(f(\emptyset))$. 
Equation (\ref{Fstar})
defines an operation ``$*$" taking us from a function $F:\complex \to
\complex$ to a function $F^*:\cM_n \to \cM_n$.  Furthermore ``$*$" is a
homomorphism under composition; viz.,
\begin{align}\label{composition}
(FG)^*=F^*G^*,
\end{align}
for any two functions $F,G :\complex \to \complex$.

We can readily extend the definition (\ref{Fstar}) of the ``$*$''
operation to operate on functions $F$ with several complex variables, e.g., 
if $F:\mathbb{C}^{2}\rightarrow \mathbb{C}$, and we have two  $\bracket$s $f$ and $g$, then
\begin{align}
F^*(f,g)(a)=\partial_a F(f,g).
\end{align}
If we now choose $G(f,g)=fg$, we can use the extended definition to find  $G^*(f,g)=f*g$. 
We can furthermore apply (\ref{composition}) to obtain
\begin{align}
\label{hom}
F^*(f*g)=(F(fg))^*.
\end{align}
Thus ``$*$" enables us to carry over maps of complex numbers to $\cM_n$,
preserving all their properties. For instance, if $F(f)=f^{-1}$, we
obtain an inverse in $\cM_n$, defined whenever $f(\emptyset) \ne 0$ by
\begin{align}
\label{inverse} (f^{-1*})(a)&=\partial_a(f^{-1}),
\end{align}
where $f^{-1}$ is the multiplicative inverse $1/f(a)$. Thus
$f*f^{-1*}=1^*$. Two other operations are
\begin{align}
\label{lg} (\log^* f)(a)&=\partial_a \log(f),\\
\label{xp} (\exp^* f)(a)&=\partial_a \exp(f),
\end{align}
where ``$\log$'' here and elsewhere means $\log_e$. From
(\ref{composition}) we deduce that $\exp^*(\log^*)=1^*$, and from
(\ref{hom}) that
\begin{equation}
\label{logprod}
\log^*(f*g)=\log^* f+\log^* g,
\end{equation} 
and $\log^*(f^{-1*})=-\log^*f$. 

The prescription for the inverse, (\ref{inverse}), can easily be
turned into an explicit rule by formal differentiation. One finds
\begin{align}
\label{inverse0} f^{-1*}(\emptyset)&=1/f(\emptyset),\\
\label{inverse1} f^{-1*}(1)&=-f(1)/f(\emptyset)^2,\\
\label{inverse2} f^{-1*}(1,2)&=-f(1,2)/f(\emptyset)^2+2f(1)f(2)/f(\emptyset)^3,
\end{align}
and in general
\begin{equation}
\label{explicitinverse}
f^{-1*}(a)=\sum_{p \in \lambda(a)} \frac{|p|!(-1)^{|p|}}{f(\emptyset)^{|p|+1}} \Pi_{c \in p} f(c),
\end{equation}
the sum being taken over all partitions $\lambda(a)$ of $a$, with
$|p|$ denoting the number of parts (subsets) of the partition
$p$. This follows at once from the combinatorial version of Fa\'{a} di
Bruno's rule \cite{Faa} (and the moment algebra acquires an added grace by this use of the
theorem of a beatified mathematician). It is remarkable
that the simple definition in (\ref{Fstar}) leads to such a
rich structure. 

Similarly, Fa\'a di Bruno's rule  gives
\begin{equation}
\label{log-explicit} \log^* f(a)=\sum_{p\in\lambda(a)} (|p|-1)!(-1)^{|p|-1} \frac{\Pi_{c\in p} f(c)}{f(\emptyset)^{|p|}}, \quad a\neq \emptyset,
\end{equation}
which is well-defined for $\bracket$s $f$ such that $f(\emptyset) \neq
0$. From (\ref{lg}) and (\ref{log-explicit}) we find
\begin{align}
\label{log-emptyset} \log^{*}f(\emptyset) &= \log f(\emptyset),\\
\label{log1} \log^* f(1)&=f(1)/f(\emptyset),\\
\label{log12} \log^* f(1,2)&=f(1,2)/f(\emptyset)-f(1)f(2)/f(\emptyset)^2.
\end{align}
We can make a similar calculation for $\exp^{*}$ and find
\begin{equation}\label{exp-explicit}
\exp^{*}f(a) = e^{f(\emptyset)}\sum_{p\in \lambda(a)}\Pi_{c\in p}f(c).
\end{equation}

We shall refer to $\log^*f$ as the \emph{cumulant} of the $\bracket$
$f$, and $\exp^* f$ as its \emph{anticumulant}. Cumulants were
originally introduced in the context of statistics and probability
theory \cite{KendallStuart77}. To formulate these ``classical"
cumulants within this framework we let $X_{1},\ldots, X_{n}$ be a
collection of random variables, and define an $\bracket$ $f$ by $f(a)
= \langle \Pi_{j\in a}X_{j}\rangle$, for each subset $a\subseteq
\{1,\ldots,n\} = \Omega_{n}$, where $\langle \cdots\rangle$ is the
expectation value of the product of the random variables. (We also
take $f(\emptyset) = \langle 1 \rangle=1$.) Then
\begin{align}\label{classical}
\log^{*}f(a)
\equiv \log^{*}\langle \Pi_{j\in a}X_{j}\rangle
\end{align}
is precisely the classical joint cumulant \cite{Royer} of the random
variables $\{X_{j}\}_{j\in a}$, and $\exp^*f(a)$ is their classical
anticumulant.  This justifies us in using the same terminology for
$\log^*f$ and $\exp^* f$ in the more general situation where $f$ is an
arbitrary $\bracket$.

We now use the moment algebra to rederive some properties of
cumulants in this more general setting. Suppose that an $\bracket$ $f$
satisfies 
\begin{equation}
f(c) = f(c\cap A)f(c\cap B), \quad \forall c\subseteq \Omega_n,
\end{equation}
where $\{A,B\}$ is a bipartition of $\Omega_n$, i.e., $A\cap B =
\emptyset$ and $A\cup B = \Omega_n$. We say that $f$ \emph{factorizes}
with respect to $\{A, B\}$ on $\Omega_n$. Then $\log^* f(c)=0$ for any
nonempty $c \subset \Omega_n$ such that $c \not\subseteq A$ or $c
\not\subseteq B$, and in particular $\log^* f(\Omega_n)=0$. In other
words the cumulant of factorizing $\bracket$s vanishes. To see this,
we note that we can write $f=f_A*f_B$, where $f_A(c) = f(c)$ if $c
\subset A$ and $f_A(c) = 0$ otherwise, and $f_{B}(c) = f(c)$ if $c
\subset B$ and $f_B(c) = 0$ otherwise. Then
\begin{align}
\log^* f(c)=\log^*f_A(c)+\log^*f_B(c),
\end{align}
and this is zero for any $c$ that is not either entirely in $A$ or
entirely in $B$, since then there is some element $a$
of $c$ with $a \notin A$, and (\ref{log-explicit}) shows that
$\log^*f_A$ must vanish; but this is also true of $\log^*f_B$.
It can be shown that this property characterizes cumulants
\cite{Percus75,Simon79}. 

We mention here a property of cumulants that we will
frequently make use of. Given $\alpha \in \complex$, define the {\em
scalar} $\alpha^*$ by $\alpha^*(\emptyset)=\alpha$, $\alpha^*(c)=0$ if
$c \ne \emptyset$. Thus in the special case $\alpha=1$, $\alpha^*$ is
what we have previously called $1^*$, so our useage is consistent. Now
$f*\alpha^*(c)=(\alpha f)(c)$. Thus the convolution product with
$\alpha^*$ is equivalent to scaling the value of $f$ for all subsets
of $\Omega_n$ by the factor $\alpha$. Note also that
$\log^*\alpha^*(c)=0$, for all $c \ne \emptyset$. Thus, by (\ref{logprod}),
\begin{equation}
\label{scale_invariance}
\log^* (f*\alpha^*)(c) = \log^*\alpha^*(c)+\log^*f(c)=\log^*f(c),
\end{equation}
for $c \ne \emptyset$, so scaling $f$ by a constant factor leaves the
cumulant unchanged on all non-empty subsets of $\Omega_n$.

The comparison between $\cM_n$ and complex analysis is further
strengthened by the existence of the power series
\begin{align}
\label{power} \log^* (1^*+f)=f-\frac{f*f}{2}+\frac{f*f*f}{3}- \ldots,
\end{align}
which converges whenever $|f(\emptyset)| < 1$. When applied to
  the empty set this becomes
\begin{align}\label{emptyset-series}
\log^* (1^*+f)(\emptyset)=f(\emptyset)-\frac{f(\emptyset)^2}{2}+\frac{f(\emptyset)^3}{3}- \ldots,
\end{align}
which is the familiar power series for the complex function
$\log(1+f(\emptyset))$, as it should be according to
(\ref{log-emptyset}).  When (\ref{power}) is applied to sets other
than $\emptyset$, checking its validity is a pleasant exercise. For
instance, using the product rule (\ref{explicit-product}) to evaluate
the repeated convolutions gives
\begin{align*}
\log^* (1^*+f)(1,2)&=f(1,2)-[f(1,2)f(\emptyset)+f(1)f(2)]+[f(1,2)f(\emptyset)^2+2f(1)f(2)f(\emptyset)] + \ldots\\
&=\frac{f(1,2)}{1+f(\emptyset)}-\frac{f(1)f(2)}{(1+f(\emptyset))^2},
\end{align*}
which one sees is the correct expression if one compares it with
(\ref{log12}) and bears in mind the definition of $1^*$ by
(\ref{identity}).

As a final ingredient, we define
\begin{align}
(\partial_i^* f)(a)=\partial_a(\partial_i f), \mbox{ if } i \notin a.
\end{align}
(The annoying restriction $i \notin a$ can be removed by means of the
multiset formalism in Section \ref{multisets}.) This operation
  behaves like a partial derivative; for instance, we have
  $\partial_i^* \log^* f(a)=(\partial_i^* f) * f^{-1*}(a)$. (In the
  case of the $\bracket$ defined by (\ref{phi-map}), $\partial_i$
  corresponds to the operation $\phi \to \partial \phi/\partial
    x_i$.) It can also be interpreted as a sort of ``raising''
operator, taking $f(a)$ to $f(a \cup i)$. This lets us immediately
derive $\log^* (1^*+f)(a)$ for any $a$ from $\log^*
(1^*+f)(\emptyset)$, which, as we have just seen, is the complex
function expansion (\ref{emptyset-series}). More generally, we can
derive a moment algebra equivalent from any complex function
power series.




\section{\label{weakmeasurements}Weak measurements}

Weak measurement, developed by Aharonov and his colleagues,
\cite{AharonovRohrlich05, AAV88}, is a strategy for extracting
information from a quantum system $\cS$, and has a number of distinct
ideas behind it:
\begin{enumerate}[(i)]
\item A measuring system $\cP$ is weakly coupled to the original
system $\cS$, and $\cP$ is then uncoupled and measured. This allows
some limited information about $\cS$ to be gained with little disturbance to
$\cS$.
\item By repeating the entire procedure many times, with the
system identically prepared on each occasion, the noisy
information about $\cS$ obtained by the weak coupling of $\cP$ can
be averaged to give a definite answer.
\item The system $\cS$ can be both preselected and {\em postselected},
the latter meaning that the procedure concludes with a measurement on
$\cS$, and only if a particular outcome is obtained are the data
included in the averaging process.
\end{enumerate}
These ideas combine constructively and enable one to probe a
system in new ways. In particular, one can express a weak measurement
result in terms of a quantity called the {\em weak value}, akin to the
standard expectation value. This depends upon both the pre- and
post-conditioning states, and can take very curious-seeming values,
that can nevertheless be shown to have a natural physical meaning
\cite{AAV88}.

In the standard weak measurement setup the pointer system
$\mathcal{P}$ has a continuous degree of freedom, like the position of
a single particle.  Another common assumption is that this pointer
particle initially is in a pure state $\phi$ for which the
expectation value of both the position and momentum observables are
zero. Hence, the probability density $|\phi(x)|^{2}$ for finding the
particle at position $x$ is a gaussian centered at zero.  Furthermore
it is assumed that the coupling between the pointer and system is
given by the impulsive Hamiltonian $H=\gamma p\otimes A\delta(t)$,
where $p$ is the momentum operator on $\cP$ and $A$ is any Hermitian
operator on $\cS$, and where $\delta(t)$ is the delta distribution
centered at time $t=0$. Suppose the system is initially in state
$\ket{\psi_i}$ and is postselected in state $\ket{\psi_f}$. The state
of the pointer after the interaction and postselection is
proportional to $\langle \psi_{f}|\exp(-i\gamma p\otimes
A)|\phi\rangle|\psi_{i}\rangle$. After the interaction, the pointer
position $q$ is measured, and by averaging over many repeats one
obtains the expectation value to any desired accuracy. 

To model the weakness of the interaction we expand the resulting
expectation value up to the first order in the interaction parameter
$\gamma$, resulting in \cite{AAV88}
\begin{align}\label{qclassic}
\langle q \rangle=\gamma Re A_w,
\end{align}
where $A_w$ is the {\em weak value} of the observable $A$ defined by
\begin{equation}
\label{weakvalue}
A_w=\frac{\bra{\psi_f}A\ket{\psi_i}}{\braket{\psi_f}{\psi_i}}.
\end{equation}
This basic type of weak measurement can readily be generalised. The
pointer can be an arbitrary quantum system. We do not necessarily have
to have a continuous degree of freedom. The Hilbert space could be
finite-dimensional; e.g., the pointer could be the spin degree of
freedom of a particle. The coupling can be $H=\gamma
\delta(t)s\otimes A$, where $s$ is now any Hermitian operator, and
likewise one can measure any Hermitian operator $r$ on $\cP$. We
further allow the initial pointer state $|\phi\rangle$ to be
arbitrary. Then (\ref{qclassic}) becomes \cite{GraemeCumulants}
\begin{align}\label{firstxi}
\langle r \rangle =\langle r \rangle_{\phi} +\gamma Re(\xi A_w),
\end{align}
where 
\begin{align} \label{xi1}
\xi=-2i\left(\langle rs \rangle_{\phi}-\langle r \rangle_{\phi} \langle s \rangle_{\phi} \right)
\end{align}
and $\langle r \rangle_{\phi}=\bra{\phi} r \ket{\phi}$. We can
conclude that, by measuring the expectation value of the observable $r$
on the pointer, we can obtain the weak value $\langle A\rangle_{w}$ on
the system.

A different direction of generalisation is to weakly measure several
operators $A_k$ either simultaneously
\cite{ReschSteinberg,Resch,Lundeen} or sequentially \cite{MJP07}. In
the latter case, one couples pointers at successive times $t_k$ via
the Hamiltonians $H_{k}= \gamma_k \delta(t-t_{k})s_k\otimes A_k$, and
assumes that the system evolves by unitaries $U_k$ between these
times. One can then calculate the moment $\langle r_1 \cdots r_n
\rangle$, i.e. the expectation of the product of pointer
measurements. It turns out that this can be expressed in terms of {\em
  sequential weak values} given by
\begin{align}
\label{seqweak}
(A_n,\ldots, A_1)_w=\frac{\bra{\psi_f}U_{n+1}A_nU_n\ldots A_1U_1\ket{\psi_i}}{\bra{\psi_f}U_{n+1}U_n \ldots U_1\ket{\psi_i}}.
\end{align}
The expression for $\langle r_1 \cdots r_n \rangle$ in terms of
sequential weak values is horrendously complicated
\cite{GraemeCumulants}. It undergoes a striking simplification,
however, if one looks at cumulants. We note that $f(a)\equiv \langle
\Pi_{j\in a}r_{j}\rangle$ is an $\bracket$, since it is well defined
for every subset $a$, if we add the assumption that $f(\emptyset) =
1$. We will denote the cumulant $\log^{*}f(a)$ by $\log^*\langle
\Pi_{j\in a}r_{j} \rangle$.

We can also define another $\bracket$ using sequential weak values on
subsets as
\begin{equation}
A_{w}(a) = \frac{\bra{\psi_f}U_{n+1}F_{n}U_n\cdots
F_{2}U_{2}F_{1}U_{1}\ket{\psi_i}}{\bra{\psi_f}U_{n+1}U_n \cdots
U_1\ket{\psi_i}},\quad F_{j} = \left\{\begin{matrix} A_{j} &
\textrm{if} & j\in a\\ \hat{1} & \textrm{if} & j\notin
a.\end{matrix}\right.
\end{equation}
If we compare this with (\ref{seqweak}) we see that we only insert
the operator $A_{j}$ if $j\in a$.  As an explicit example, consider
the case of four pointers and the sequential weak value $A_{w}(2,4)$,
which would be
\begin{equation}
A_{w}(2,4) = \frac{\bra{\psi_f}U_{5}A_{4}U_{4}U_{3}A_{2}U_{2}U_{1}\ket{\psi_i}}{\bra{\psi_f}U_{5}U_{4}U_{3}U_{2}U_{1}\ket{\psi_i}}.
\end{equation}   
We will refer to $\log^{*}A_{w}$ as the \emph{sequential weak value
cumulant}.

The following theorem was proved in
\cite{GraemeCumulants}:
\begin{theorem}[Cumulant theorem]\label{old-theorem}
To the lowest joint order in the variables $\gamma$, 
\begin{align}\label{main-result}
\log^* \langle \Pi_{j\in \Omega}r_{j}\rangle= (\Pi_{j \in \Omega} \gamma_{j})  Re \left\{ \xi
\log^* A_w(\Omega)\right\},
\end{align}
where $\xi$ is given by
\begin{align}
\label{xi} \xi=2(-i)^{|\Omega|} \left( \Pi_{j\in \Omega}
 \langle r_{j}s_{j} \rangle_{\phi_{j}}-\Pi_{j\in \Omega} \langle r_{j} \rangle_{\phi_{j}} \langle s_{j} \rangle_{\phi_{j}} \right).
\end{align}
\end{theorem}
As explained above, the parameters $\gamma_{j}$ signify
the interaction strength between the system and the pointers, and we
obtain the weak measurement scenario by expanding the cumulant
$\log^{*}\langle \Pi_{j\in \Omega} r_{j}\rangle$ in the strength
parameters $\gamma$, keeping only the lowest order coefficients. This
is also the sense in which the equality in equation
(\ref{main-result}) is to be interpreted: as an equality up to the
lowest orders.  Note furthermore that the ``joint order" of
$\gamma_{2}\gamma_{1}$ is $2$, for $\gamma_{5}\gamma_{3}\gamma_{2}$ it
is $3$, etc. Hence, in the above theorem the lowest (nonzero) joint
order in the expansion is $|\Omega| = n$.  By this theorem we can
conclude that the weak measurement setup can be used to measure the
sequential weak value cumulant on the system, since in the weak limit
the joint cumulant of the pointer observables is simply related to the
sequential weak value cumulant on the system. This can thus be
regarded as a natural generalization of the weak measurement scenario
in the single pointer case, where the expectation value of the single
pointer observable corresponds to the weak value on the system.

In the case of simultaneous weak measurement
\cite{ReschSteinberg,Resch,Lundeen}, one couples $n$ pointers at a
time $t_0$ via the Hamiltonian $H_{k}= \sum_k \delta(t-t_0)\gamma_k
s_k\otimes A_k$, and the {\it simultaneous weak values} are given by
\begin{align}
\label{simweak}
(A_1,\ldots, A_n)_{ws}=\frac{1}{n!}\sum_\pi\frac{\bra{\psi_f}A_{\pi(n)}\ldots A_{\pi(1)}\ket{\psi_i}}{\braket{\psi_f}{\psi_i}},
\end{align} 
where $\pi$ runs over all permutations of $1, \ldots, n$. Again, one
can define an $\bracket$ and a cumulant denoted by $\log^*
A_{ws}$. There is an analogue to Theorem (\ref{old-theorem}) for the
simultaneous case \cite{GraemeCumulants}:

\begin{theorem}[Cumulant theorem for simultaneous
measurement]\label{old-simultaneous-theorem} If $\langle
\Pi_{j\in \Omega}r_j\rangle$ is the expected value of the product of $n$ pointers
in a simultaneous weak measurement, then, to the lowest joint order in
the variables $\gamma$,
\begin{align}\label{main-resultweak}
\log^* \langle \Pi_{j\in \Omega}r_j\rangle= (\Pi_{j\in \Omega} \gamma_j)  Re \left\{ \xi
\log^* A_{ws}(\Omega)\right\},
\end{align}
where $\xi$ is as in Theorem (\ref{old-theorem}).
\end{theorem}




\section{\label{newproof}A new cumulant theorem}

The original proof of Theorem \ref{old-theorem} was by no means
transparent. We will show how the moment algebra setting
allows a better proof. However, we begin by considering a different
way of gathering information from the pointers, where the
corresponding theorem can be proved more simply (and in section
\ref{oldtheorem} we prove the original result, which requires an
extra flourish of cumulant technology).

The assumption behind Theorem \ref{old-theorem} is that the 
moment $\langle \Pi_{j\in a} r_{j} \rangle$ for a subset $a$ of
the total collection of pointers $\Omega$ is obtained by coupling just
that set of pointers to the system: in other words, to obtain $\langle
\Pi_{j\in a} r_{j} \rangle$, one does an experiment in which just the
pointers in $a$ are coupled, so a separate experiment is needed for
each subset. The alternative approach that we now adopt is to suppose
that a single experiment is carried out in which {\it all} the
pointers are coupled and measured, but a subset of the measurement
results is used for calculating each moment.

Another way to put this is to say that in the present case we have the
total time dependent Hamiltonian $H(t) = \hat{1}_{\mathcal{P}}\otimes
H_{\mathcal{S}}(t) +
\sum_{k=1}^{n}\gamma_{k}\delta(t-t_{k})s_{k}\otimes A_{k}$, where
$H_{\mathcal{S}}(t)$ is the Hamiltonian that generates the unitary
evolution between the coupling with the pointers. Hence,
$H_{\mathcal{S}}(t)$ generates the sequence of unitary operators
$U_{k}$. In Theorem \ref{old-theorem}, however, we had a separate
experiment with a separate Hamiltonian $H_{a}(t) =
\hat{1}_{\mathcal{P}}\otimes H_{\mathcal{S}}(t) + \sum_{l\in
a}\gamma_{k}\delta(t-t_{l})s_{l}\otimes A_{l}$ for each subset $a$.

We shall see that the joint coupling of all pointers leads to a result
that differs from Theorem \ref{old-theorem}. This comes about because
of the perturbation of the system by those pointers whose readings are
not being used to calculate $\langle \Pi_{j\in a} r_{j}
\rangle_\eta$. When we measure a single pointer expectation $\langle
r_k \rangle_\eta$ this is given, up to first order, by the standard
expression $\gamma_k Re (\xi_k A_w)$, the coupling of other pointers
producing perturbations that are only of second order in the
$\gamma$'s. However, if we measure $\langle r_ir_j \rangle_\eta$ in
the presence of other pointers, the standard second order expression
is no longer correct, other second order terms being introduced by
other pointers. Despite these complications, we still obtain a
succinct relationship between cumulants of pointers and weak values,
as we shall see shortly.

To state the theorem, let $\eta$ be the state of all the pointers
after coupling and the postselection on the system. (We define $\eta$
more precisely in (\ref{sigma}) below.) The moment of the
pointer observables in the subset $a$ is $\langle \Pi_{j\in a}
r_{j}\rangle_\eta = \tr(\eta\Pi_{j\in a} r_{j})$, where the subscript
$\eta$ indicates that all the pointers are coupled.

\begin{theorem}[Cumulant theorem with all pointers coupled]\label{newtheorem}
Suppose all $n$ pointers are coupled sequentially. Let $a$ be a subset of the finite  
collection of pointers $\Omega$. To the lowest joint
order in the variables $\gamma$,
\begin{align}\label{new-result}
\log^* \langle \Pi_{j\in a}r_{j} \rangle_\eta= (\Pi_{j\in a}\gamma_{j})
Re \left\{ \xi \log^*A_w(a)\right\},
\end{align}
where $\xi$ is given by
\begin{align}
\label{new-xi} \xi=2(-i)^{|a|} \prod_{j\in a}(\langle r_js_j \rangle_{\phi_{j}}- \langle r_j \rangle_{\phi_{j}}
 \langle s_j \rangle_{\phi_{j}} )= 2(-i)^{|a|} \prod_{j\in a} \log^*\langle r_j,s_j \rangle_{\phi_{j}}.
\end{align}
\end{theorem}
This theorem allows us to consider the
cumulants of subsets of the collection of pointers. This is a slight
generalization of the theorems in
Sec. \ref{weakmeasurements}, and is more in line with the
characterization in Sec. \ref{convalgebra} of cumulants as mappings
from $\bracket$s to $\bracket$s. Just as the $\bracket$ $\langle
\Pi_{j\in a}r_{j} \rangle_\eta$ is defined for every subset of
pointers $a$, so is the corresponding cumulant $\bracket$
$\log^{*}\langle \Pi_{j\in a}r_{j} \rangle_\eta$.

\begin{proof}
First, we establish some notation. Let $\ket{\phi_k}$ be the initial
state of pointer $k$ in state space $\cH_k$, and denote the total
state space $\cH_1 \otimes \ldots \otimes \cH_n$ of pointers by
$\mathcal{P}$, where we assume $|\Omega| = n$. Let $\ket{\psi_i}$ be the initial state of the system
$\cS$, and $\ket{\psi_f}$ the postselected final state. Define
projectors on $\cS \otimes \mathcal{P}$ by
\begin{align*}
&P_i=\proj{\psi_i} \otimes \left(\proj{\phi_1} \ldots \proj{\phi_n}\right) ,\\
&P_f=\proj{\psi_f} \otimes I_{\mathcal{P}}.
\end{align*} 
Let $r_k$ be the pointer variable measured on pointer $k$ and
$H_k=s_k\otimes A_k$ be the coupling between pointer $k$ and the
system $\mathcal{S}$. For any state $\rho$ on $\cS \otimes
\mathcal{P}$, define
\begin{align}\label{L-definition}
L_k \rho=-i[H_{k},\rho],
\end{align}
so
\begin{align*}
e^{\gamma_kL_k}\rho= e^{-i\gamma_{k}H_k}\ \rho \ e^{i\gamma_{k}H_k}.
\end{align*}
Apart from the coupling to the pointer we also assume that the system
$\mathcal{S}$ has a (possibly time dependent) Hamiltonian of its own,
generating a unitary evolution between the couplings to the pointers.
We let $U_{k}$ be the unitary operator on $\mathcal{S}$ that gives the
evolution between the coupling with pointer $k-1$ and the
coupling with pointer $k$, and let $U_{0}$ be the unitary that maps the
initial state to the application of the first pointer. We then define
\begin{equation}
\mathcal{U}_{k}\rho = (\hat{1}_{\mathcal{P}}\otimes U_{k})\rho (\hat{1}_{\mathcal{P}}\otimes U_{k})^{\dagger}
\end{equation}
to be the corresponding unitary channel.  The state of the pointers after
the interaction and postselection on the system can thus be written as
\begin{equation}\label{sigma}
\eta=\frac{\tr_{\mathcal{S}}(P_f \mathcal{U}_{n+1}e^{\gamma_{n} L_{n}}\mathcal{U}_{n}\cdots e^{\gamma_{1} L_{1}}\mathcal{U}_{1}P_i)}{\tr (P_f \mathcal{U}_{n+1}e^{\gamma_{n}L_{n}}\mathcal{U}_{n}\cdots e^{\gamma_{1} L_{1}}\mathcal{U}_{1}P_i)}.
\end{equation}
Let us define the $\bracket$ $\langle \Pi_{j\in a} r_{j}\rangle_\eta =
\tr \left[(\Pi_{j\in a} r_{j})\eta\right]$.  Suppose now that we
replaced one of the observables $r_{k}$ in $ \langle \Pi_{j\in a}
r_{j}\rangle_\eta $ with the identity operator $\hat{1}_{k}$. We would
then obtain a new $\bracket$ $u_{k}(a) = \langle \Pi_{j\in a\setminus
  \{k\}} r_{j}\rangle_\eta$. This $\bracket$
factorizes on $a$ with respect to the partition
$(\{k\},a\setminus\{k\})$, i.e., for all $b\subset a$ we have
$u_{k}(b) = u_{k}\big(b\cap(a\setminus
\{k\})\big)u_{k}(b\cap\{k\})$. Thus we can conclude that $\log^{*}
u_{k}(a)=0$.  If we combine this with the multilinearity of the
cumulant in the pointer observables we find that
\begin{equation}
\log^{*}\langle \Pi_{j\in a}( r_{j}-\langle r_{j}\rangle_{\phi_{j}})\rangle_\eta  = \log^{*}\langle \Pi_{j\in a} r_{j}\rangle_\eta.
\end{equation}
Define
\begin{equation}
\label{newtildez} v(a)=\frac{\tr\left[P_f \{\Pi_{l \in a} (r_l- \langle r_l \rangle_{\phi_{l}}\hat{1}_l)\}\mathcal{U}_{n+1} e^{\gamma_{n} L_{n}}\mathcal{U}_{n}\cdots e^{\gamma_{1} L_{1}}\mathcal{U}_{1} P_i \right]}{|\langle\psi_{f}|U_{n+1}\cdots U_{1}|\psi_{i}\rangle|^{2}}.
\end{equation}
The denominators of (\ref{newtildez}) and (\ref{sigma}) do not
depend on $a$, and are thus scalars in $\cM_n$. The scale-invariance
of the cumulant, (\ref{scale_invariance}), thus yields
\begin{equation}
\label{newlogz} \log^*\langle \Pi_{j\in a}( r_{j}-\langle r_{j}\rangle_{\phi_{j}})\rangle_\eta =\log^* v.
\end{equation}
We shall now expand the above cumulant to the lowest joint order in
the parameters $\gamma_{j}$. Let us write $\partial^\gamma_b$ for the
derivative with respect to the variables $\gamma_j$ with labels in the
set $b$, i.e.,
\begin{equation}
\partial^\gamma_b \equiv \frac{\partial}{\partial \gamma_{b(|b|)}}\cdots \frac{\partial}{\partial \gamma_{b(1)}}.
\end{equation}
(These ``proper" derivatives $\partial^\gamma_b$ should not be
confused with the formal derivatives $\partial_{a}$ introduced in
section \ref{convalgebra}.)  We next prove the following
\begin{equation}
\label{expansion}
\partial^\gamma_b \log^* v(a)|_{\gamma=0}=\left\{\begin{matrix}0  & \textrm{if} & a\setminus b \neq \emptyset \\
 \log^* w(a) & \textrm{if}& b = a,
\end{matrix}\right. 
\end{equation}
where
\begin{equation}
\label{wdef}
w(a) = \partial^{\gamma}_{a}v(a)\vert_{\gamma = 0} 
= \frac{\tr(P_f \mathcal{U}_{n+1}W_n\mathcal{U}_{n}\cdots W_1 \mathcal{U}_{1} P_i )}{|\langle\psi_{f}|U_{n+1}\cdots U_{1}|\psi_{i}\rangle|^{2}},\quad W_j = \left\{\begin{matrix} (r_{j}-\langle r_{j}\rangle_{\phi_{j}}\hat{1}_{j})L_{j} & j\in a,\\
\hat{1} & j\notin a,
\end{matrix}
\right.
\end{equation}
and $\gamma=0$ means $\gamma_j=0$, for $1 \le j \le n$. Equation
(\ref{expansion}) tells us that the first potentially nonzero
expansion coefficient of the cumulant $\log^* v(a)$ can itself be
regarded as a cumulant, but of the new $\bracket$ $w$.  This method of
regarding the expansion coefficient of a cumulant as a cumulant in its
own right is a technique that we will use repeatedly.

To prove (\ref{expansion}), let us first suppose that $a\setminus b\neq \emptyset$. Hence, there
must be some element $j\in a$ such that $j\notin b$. Recall the
expression for the cumulant in terms of partitions,
(\ref{log-explicit}), and suppose $p$ is a partition of $a$. Then
$\partial_{b}^{\gamma}\Pi_{c\in
p}v(c)|_{\gamma=0}=0$, since for that $c \in p$
that contains $j$ the derivative $\partial/\partial
\gamma_{j}$ is not applied to $v(c)$, and consequently this
term contains $r_{j}-\langle r_{j} \rangle
\hat{1}_{j}$ but not $L_{j}$ and must therefore
vanish. This gives the first part of (\ref{expansion}). For the
case $b = a$, we again use the fact that, for any $j \in a$, a term
containing $r_j-\langle r_j \rangle \hat{1}_j$ but not $L_j$ must
vanish, to find that $\partial_{a}^{\gamma}\Pi_{c\in
p}v(c)|_{\gamma = 0} = \Pi_{c\in
p}\partial_{c}v(c)|_{\gamma=0}$, for any partition $p$ of
$a$. The statement in (\ref{expansion}) follows.

It remains to evaluate the cumulant $\log^* w(a)$ .  Let us rephrase
the definition of $L_j$ in (\ref{L-definition}) as
\begin{align}
\label{L-def} L_j= L_{j}^{\textrm{left}} + L_{j}^{\textrm{right}},\quad L_{j}^{\textrm{left}}(\rho) = (-is_j \otimes A_j)\rho,\quad L_{j}^{\textrm{right}}(\rho) = \rho(is_j \otimes A_j),
\end{align}
so the subscript `left' (`right') indicates which side the operator is
applied to.  We can write
\begin{equation}
w(a) = \sum_{c_{1},c_{2}}w_{c_{1},c_{2}}(a),
\end{equation}
where the sum is over all ordered bipartitions $(c_{1},c_{2})$ of $a$,
and where $w_{c_{1},c_{2}}(a)$ is defined as in (\ref{wdef}) but with
$L_j$ replaced by $L_j^{\textrm{left}}$ when $j \in c_1$, and by
$L_j^{\textrm{right}}$ when $j \in c_2$. 
By using the fact that $L_j^{\textrm{left}}$ and 
$L_j^{\textrm{right}}$ act independently, and that $P_{i}$ is a projector onto pure product states, one can show that $w_{c_{1},c_{2}}$ 
factorizes on $a$ with respect to the partition $\{c_{1},c_{2}\}$. Thus $\log^*w_{c_{1},c_{2}}(a)=0$ except when either
$c_1=\emptyset$ or $c_2=\emptyset$; so
$\log^*w(a)=\log^*w_{a,\emptyset}(a)+\log^*w_{\emptyset,a}(a)$. Direct
calculation shows that $\log^*w_{a,\emptyset}=(-i)^{|a|} \{\Pi_{k\in
a} (\langle r_ks_k \rangle- \langle r_k \rangle \langle s_k \rangle
)\} A_w(a)$, and $\log^*w_{\emptyset,a}$ gives the complex
conjugate. The theorem follows.
\end{proof}




\section{\label{simulevol}Simultaneous weak measurement with system evolution}
The last section focussed on sequential weak measurement. In the case
of simultaneous weak measurement, it is assumed
\cite{ReschSteinberg,Resch,Lundeen} that the Hamiltonian has the form
$H_{k}= \sum_k \delta(t-t_0)\gamma_k s_k\otimes A_k$, and the
evolution of the system, which occurred between coupling of pointers
in sequential weak measurement, can be ignored here since the coupling
occurs impulsively and simultaneously for all pointers, and any
evolution before or after the coupling can be incorporated into the
initial and final states by writing $\ket{\tilde
\psi_i}=U_1\ket{\psi_i}$, $\bra{\tilde \psi_f}=\bra{\psi_f}U_2$.

But suppose the coupling occurs over a finite time. There is
essentially no change in the analysis if we assume a Hamiltonian
$H_{k}= f(t)\sum_{k=1}^n g_k s_k\otimes A_k$, where $f(t)$
defines some time-course for the coupling. However, there is now the
possibility of having the system evolve while the coupling is
occurring. Suppose this evolution is given by the Hamiltonian
$H_\cS$. Then the total Hamiltonian is $H=\hat{1}_{\mathcal{P}}\otimes
H_\cS+f(t)\sum_{k=1}^n g_k s_k\otimes A_k$. We assume here that
$f(t)$ is constant on a time interval of length $\tau$, and find it
convenient to let our expansion parameters $\gamma_{k}$ be the total
strength of the interaction between the system and pointers, given by
$\gamma_{k} = \int g_{k}f(t)dt = \tau g_{k}$. For times $0\leq t\leq \tau$ we can thus
write the Hamiltonian as
\begin{equation}
H=\hat{1}_{\mathcal{P}}\otimes H_\cS+\sum_{k=1}^n \gamma_k
\frac{s_k}{\tau}\otimes A_k.
\end{equation}

As before we assume the system and the pointers start in a total
product state, and that after the time interval $\tau$ we postselect
the system in $|\psi_{f}\rangle$. The state of the pointers after the
postselection is
\begin{equation}
\label{etaprim}
\sigma = \frac{\tr_{\mathcal{S}}(P_f e^{-i\tau H}P_i e^{i\tau
H})}{\tr(P_f e^{-i\tau H}P_i e^{i\tau H})},
\end{equation}
and we measure the observables $r_{1},\ldots,r_{n}$, respectively, on
each pointer, and calculate the corresponding cumulants
$\log^{*}\langle \Pi_{j\in a} r_{j}\rangle_{\sigma}$.  Now, for
$a = \{a_1,\ldots, a_k\}$, let
\begin{align}
\label{evolution-weak-value} (A_{a_{k}}, \ldots, A_{a_{1}})_w[\tau_{k+1}, \ldots ,\tau_1]
=\frac{\bra{\psi_f}e^{-iH_\cS \tau_{k+1}}A_{a_{k}}e^{-iH_\cS \tau_k}\cdots A_{a_1}e^{-iH_\cS \tau_1}\ket{\psi_i}}{\bra{\psi_f}e^{ -i\tau H_\cS}\ket{\psi_i}}
\end{align}
and define the $\bracket$
\begin{equation}
\label{mathcDdef}
\cD(a)=\frac{1}{\tau^k}\sum_{\pi} \int_{  \tau_{k+1}, \tau_{k}, \ldots, \tau_1 \geq 0    } (A_{a_{\pi(k)}}, \ldots,A_{a_{\pi(1)}})_w[\tau_{k+1}, \ldots \tau_1]\delta(\tau-\sum_{j=1}^{k+1} \tau_j ) d\tau_{1}\cdots d\tau_k d\tau_{k+1}.
\end{equation}

$\mathcal{D}$ can be regarded as an average over all possible
sequential weak values of the operators $A_{a_{1}},\ldots,
A_{a_k}$, where we take all possible rearrangements of the order
in which these operators are measured, as well as varying the time
steps between the applications.

\begin{theorem}[Weak simultaneous measurement with system evolution]\label{sim1}
Let $a$ be a subset of the finite collection of pointers $\Omega$. To the lowest joint order in the variables $\gamma$,
\begin{align}\label{simevol}
\log^* \langle \Pi_{j\in a}r_{j} \rangle_{\sigma} = (\Pi_{j\in a}\gamma_{j}) Re \left\{ \xi
\log^*\mathcal{D}(a)\right\},
\end{align}
where $\xi$ is as in Theorem \ref{newtheorem}.
\end{theorem}

To prove this, we begin with the following
\begin{lemma}\label{Dyson-lemma}
Let $a=\{a_1, \ldots, a_k\}$. Then
\begin{align}\label{dirichlet}
\partial^\gamma_a e^{\tau(Y+\sum_{j=1}^n \gamma_jX_j)}\vert_{\gamma=0} = \sum_{\pi}\int_{ \tau_{k+1},\tau_{k}, \ldots, \tau_{1}\geq 0 }  e^{\tau_{k+1}Y}X_{a_{\pi(k)}}e^{\tau_{k}Y}X_{a_{\pi(k-1)}}\cdots e^{\tau_{2}Y}X_{a_{\pi(1)}}e^{\tau_{1}Y}\delta(\tau -\sum_{j=1}^{k+1}\tau_{j})d\tau_{1}\cdots d\tau_{k}d\tau_{k+1},
\end{align}
where the sum is over all permutations $\pi$ of the set $\{1, \ldots,k\}$.
\end{lemma}
\begin{proof}
We use the Dyson series (proof: differentiate both sides with respect
to $\tau$)
\begin{align}
e^{\tau (Y+V)}=e^{\tau Y}+\int^{\tau}_0e^{(\tau-t_1)Y}Ve^{t_1Y}dt_1+\int^\tau_0\int^{t_2}_0e^{(\tau-t_2)Y}Ve^{(t_2-t_1)Y}Ve^{t_1Y}dt_1dt_2 + \ldots.
\end{align}
Putting $V=\sum_{j=1}^n\gamma_jX_j$, the only term that survives the
combined operations of differentiation by
$\partial_{a}^{\gamma}$ and setting
$\gamma=0$ is the $k$-times repeated integral
\begin{equation}
\partial_{a}^{\gamma}e^{\tau (Y+V)}|_{\gamma=0} 
= 
\sum_{\pi}\int_{0}^{\tau}\int_{0}^{t_{k}}\cdots\int_{0}^{t_{2}}e^{(\tau-t_{k})Y}X_{a_{\pi(k)}}e^{(t_{k}-t_{k-1})Y}X_{\pi(k-1)}\cdots e^{(t_{2}-t_{1})Y}X_{\pi(1)}e^{t_{1}Y}dt_{1}\cdots dt_{k}.
\end{equation}
To obtain (\ref{dirichlet}), make the change of variables
$\tau_1=t_1$, $\tau_2 = t_2-t_1, \ldots, \tau_k = t_k-t_{k-1}$,
$\tau_{k+1}=\tau-t_k$.
\end{proof}

\begin{proof}[Proof of Theorem]
  Consider the $\bracket$ $\log^* \langle \Pi_{j\in a}
  r_{j}\rangle_{\sigma}$.  Following the proof of Theorem
  \ref{newtheorem}, we can replace all the observables $r_{j}$ with
  $r_{j}-\langle r_{j}\rangle \hat{1}_{j}$ without changing the
  cumulant.  Still following Theorem \ref{newtheorem} we find that the
  lowest order term in the expansion has joint degree $|a|$ and
  corresponding expansion coefficient
  $\partial_{a}^{\gamma}\log^*\langle \Pi_{j\in a}( r_{j}-\langle
  r_{j}\rangle_{\phi_{j}})\rangle_\sigma |_{\gamma = 0} =
  \log^{*}d(a)$ with the new $\bracket$
\begin{eqnarray}
d(a) & = &\frac{\partial_{a}^{\gamma}
\tr \left\{ P_f \big( \Pi_{j\in a}(r_{j}-\langle r_{j}\rangle \hat{1}_{j})\big) e^{-i \tau H}P_i e^{i \tau H} \right\} |_{\gamma = 0 }}{|\langle\psi_{f}| e^{-i\tau H_{\mathcal{S}}}|\psi_{i}\rangle|^{2}}\nonumber\\
 & = &\frac{1}{|\langle\psi_{f}| e^{-i\tau H_{\mathcal{S}}}|\psi_{i}\rangle|^{2}} \sum_{(c_{1},c_{2})}\tr\left\{P_f \big(\Pi_{j \in a}(r_{j}-\langle r_{j}\rangle \hat{1}_{j}) \big) \big[\partial_{c_{1}}^{\gamma}e^{-i \tau H}\big] P_i \big[\partial_{c_{2}}^{\gamma}e^{i \tau H}\big] \right\}\big|_{\gamma = 0},
\end{eqnarray}
where the sum is over all ordered bipartitions $(c_{1}, c_{2})$ of
$a$. 
 Next, we apply Lemma \ref{Dyson-lemma} to both
$\big[\partial_{c_{1}}^{\gamma}e^{-i \tau
H}\big]|_{\gamma = 0}$ and
$\big[\partial_{c_{2}}^{\gamma}e^{i \tau H}\big]|_{\gamma
= 0}$, with $Y = -i\hat{1}_{\mathcal{P}}\otimes H_{\mathcal{S}}, X_{k}
= -\frac{i}{\tau}s_{k}\otimes A_{k}$, and $Y =
i\hat{1}_{\mathcal{P}}\otimes H_{\mathcal{S}}, X_{k} =
\frac{i}{\tau}s_{k}\otimes A_{k}$, respectively. We find that
$d=D*D'$, where
\begin{equation}
D(a) = (-i)^{|a|}\big(\Pi_{j\in a}(\langle r_{j}s_{j}\rangle_{\phi_{j}}-\langle r_{j}\rangle_{\phi_{j}}\langle s_{j}\rangle_{\phi_{j}})\big)\mathcal{D}(a),
\end{equation}
and where $D'$ is the complex conjugate of $D$. 
The statement of
Theorem \ref{sim1} follows from $d = D*D'$ together with the fact that
$\log^{*}d = \log^{*}D + \log^{*}D'$. 
\end{proof}

The interpretation of this theorem is quite intuitive. If there were
no evolution, i.e. $H_\cS=0$, then $\cD(a)$ would be the simultaneous
weak value, \cite{ReschSteinberg,Resch,Lundeen} and
(\ref{simweak}), given by symmetrizing over all orders of applying
operators, viz.
\begin{align}
\cD(a)=\frac{1}{n!}\sum_\pi (A_{a_{\pi(n)}}, \ldots,
A_{a_{\pi(1)}})_w,
\end{align}
where the factor $1/n!$ comes from integrating $\tau_1, \ldots,
\tau_n,\tau_{n+1} \ge 0$ with the constraint $\sum_{j=1}^{n+1}
\tau_j=\tau$. When $H_\cS$ is nonzero, we must in addition average
over episodes of evolution under $e^{-iH_\cS t}$ between application of
the $A_k$ with the lengths of all episodes summing to $\tau$.

The theorem implies that simultaneous weak measurement can be
simulated by collections of sequential weak measurements, by sampling
over permutations of the ordering of the applications of the pointers,
as well as over the time steps between the applications of the
pointers.




\section{\label{thermal}Thermal weak measurement}

In the previous section we stretched the concept of weak measurement a
little by allowing the system to evolve while the pointers are
coupled. Here we stretch it further by abandoning the notion of preselection and postselection (key
ingredients of the original weak measurement philosophy \cite{ABL64,AV91}), and
instead considering a system in
thermal equilibrium. As we will see, this thermal weak measurement
concept is, formally speaking, closely related to the simultaneous
weak measurement with system evolution considered in the previous
section.  The correspondence between these two scenarios is analogous
to that between path integrals and equilibrium systems under the
imaginary time transformation $t \leftrightarrow it$ \cite{Wenbook}.

For thermal equilibrium systems with Hamiltonian $H$, the Helmholz
free energy \cite{Landau} can be written 
\begin{equation*}
F = -\frac{1}{\beta}\log
Z(\beta) = -\frac{1}{\beta}\log \tr e^{-\beta H},
\end{equation*}
 where $\beta=1/kT$,
with $T$ being the temperature. 
When external parameters, e.g. fields, are changed infinitely slowly, 
the difference between the final and initial free energy is equal to the work performed on the system, under
the assumption that the system is kept in contact with a heat bath 
at constant temperature $T$ \cite{Landau}.  The Taylor expansion of the free
energy with respect to the external fields thus characterizes the
system's response to small changes.  This picture can be extended to
several fields $g_{j}$ coupling to the system via observables $A_{j}$, for instance with linear coupling 
\begin{align}\label{linear}
H_\cC=H_\cS + \sum_{j}g_{j}A_{j},
\end{align}
leading to a free energy
\begin{equation}
\label{freeenergy}
F = -\frac{1}{\beta}\log \tr e^{-\beta H_\cC}.
\end{equation}

We follow \cite{Percus75} and refer to the expansion coefficients of
$F$ with respect to $\gamma_{1},\ldots, \gamma_{n}$,
i.e. $\partial^\gamma_a F|_{\gamma=0}$ for a subset $a$ of the
$\gamma$'s, as \emph{generalized susceptibilities}.  We shall now
construct a weak measurement scenario where the correlation of the
pointers, as measured by the joint cumulant of the pointer
observables, turns out to be directly proportional to these
generalized susceptibilities.

Instead of letting a collection of pointers weakly interact with the
system for specific times, here we let the pointers and system
equilibrate under the assumption of weak interactions. When this
combined system has equilibrated we separate the pointers and, as
before, measure a collection of observables on them.  We
therefore consider a total Hamiltonian
\begin{align}
H=\hat{1}_{\mathcal{P}}\otimes H_{\mathcal{S}}+\sum_{j=1}^n \gamma_j
\frac{s_j}{\beta}\otimes A_j,
\end{align}
where $\gamma_k=\beta g_k$, and we assume that the system and the
pointers reach the thermal equilibrium state under this
Hamiltonian, yielding
\begin{equation}
\rho = \frac{e^{-\beta H}}{\tr e^{-\beta H}}.
\end{equation}
 The expectation for pointer measurements is the $\bracket$ $\langle
\Pi_{j\in a} r_{j} \rangle_{\rho}=\tr(\rho\Pi_{j\in a} r_{j} )$.  We
also define the $\bracket$
\begin{equation}
\label{medef}
\mathcal{E}(a) =\frac{\partial_{a}^{\gamma}\tr\{e^{-\beta H_\cS-\sum_j\gamma_jA_j}\}|_{\gamma=0}}{\beta^{|a|} \tr e^{-\beta H_\cS}},
\end{equation}
which gives, up to a normalising factor, the Taylor coefficients in
the expansion of the partition function with respect to
$\gamma$.

\begin{theorem}[Weak measurement of a system in equilibrium]\label{sim2}
Let $a$ be a subset of the finite  collection of pointers $\Omega$. To the lowest joint order in $\gamma$ we find
\begin{equation}
\label{thermalresult}
\log^* \langle \Pi_{j\in a} r_{j} \rangle_{\rho}= (\Pi_{j\in a}
\gamma_j) \xi \log^* \mathcal{E}(a)=-\beta (\Pi_{j\in a}
\gamma_j) \xi \partial^\gamma_a F|_{\gamma=0},
\end{equation}
\begin{equation}
\xi = \Pi_{j\in a}(\tr(r_{j}s_{j})-\tr(r_{j})\tr(s_{j})).
\end{equation}
\end{theorem}
This assumes that the various traces $\tr(r_{j})$, $\tr(s_{j})$, and
$\tr(r_js_j)$ are well defined, which in the case of infinite
dimensional Hilbert spaces requires them to be trace class
\cite{ReedSimon}. The theorem tells us that the joint cumulant
of the pointers is directly proportional to the generalized
susceptibility of the system.

As an application of these ideas, consider a collection of
pointers in equilibrium with a heat bath. Correlations between the
pointers will be generated by the heat bath, and these can be
characterised by local observables. Our theorem says that, if the
coupling of the pointers to the heat bath is weak, the cumulants of
these local observables will be proportional to the generalized
susceptibilities of the heat bath.

As mentioned above, there is an analogy between thermal weak
measurement and simultaneous weak measurement.  To see this, we can
use Lemma \ref{Dyson-lemma} to expand the $\bracket$ $\mathcal{E}$ in
weak values as
\begin{align}\label{Ea}
\mathcal{E}(a)=\frac{1}{\beta^k}\sum_{\sigma} \int_{\tau_{k+1},\ldots, \tau_{1} \geq 0} 
(A_{a_{\sigma(k)}}, \ldots, A_{a_{\sigma(1)}})_e[\tau_{k+1},\ldots , \tau_1]\delta(\beta-\sum_{j=1}^{k+1} \tau_j ) d\tau_{1}\cdots d\tau_{k+1},
\end{align}
where
\begin{align}
\label{equilib} (A_k, \ldots, A_1)_e[\tau_{k+1},\ldots ,\tau_1]=\frac{\tr\left[e^{-\tau_{k+1}H_\cS}A_k e^{-\tau_{k}H_\cS}\cdots A_1e^{- \tau_{1}H_\cS} \right]}{\tr\left[ e^{-\beta H_\cS} \right]}.
\end{align}
Comparing (\ref{Ea}) and (\ref{equilib}) with (\ref{mathcDdef}) and
(\ref{evolution-weak-value}), respectively, the $t \leftrightarrow it$
correspondence is clear; this allows us to carry over a large part of
the proof of Theorem \ref{sim1} to the present theorem.

\begin{proof}[Proof of Theorem]
Consider the $\bracket$  $\langle \Pi_{j\in a} r_{j} \rangle_{\rho} = \tr(\rho\Pi_{j\in a} r_{j} )$.  In the analogue of
 (\ref{newtildez}), instead of replacing $r_{j}$ by $(r_j- \langle
r_j \rangle_{\phi_j}\hat{1}_j)$, we replace it by $(r_{j}-
\tr(r_{j})\hat{1}_{j})$, the trace playing the role previously taken
by the expectation.  As before, this modification of the pointer
observables does not change the cumulants,  $\log^{*}\langle \Pi_{j\in a} r_{j} \rangle_{\rho} =  \log^{*}\langle \Pi_{j\in a}\big( r_{j} -\tr(r_{j})\hat{1}\big)\rangle_{\rho}$.
Defining 
\begin{equation}
X_{\cP} = \frac{\Pi_{j \in a}(r_j- \tr(r_j)\hat{1}_j)}{\tr e^{-\beta H_{\cS}}}
\end{equation}
we use scale-invariance, (\ref{scale_invariance}), to show
\begin{align}
\label{line1} \partial_{a}^{\gamma} \langle \Pi_{j\in a}\big( r_{j} -\tr(r_{j})\hat{1}\big)\rangle_{\rho}  |_{\gamma=0} 
&=\partial_{a}^{\gamma} \log^* \tr(e^{-\beta H}X_\cP)|_{\gamma=0} \\
\label{line2} &=\log^* \partial_{a}^{\gamma} \tr(e^{-\beta H}X_\cP)|_{\gamma=0} \\
\nonumber&=\log^* \tr(\Pi_{j \in a} s_j X_\cP)\ \partial_{a}^{\gamma}
 \tr(e^{-\beta H_\cC})|_{\gamma=0}\\ 
\nonumber&=\xi \log^* \cE(a),
\end{align}
where (\ref{line2}) follows from (\ref{line1}) by the same argument
that derived (\ref{expansion}) from (\ref{wdef}).  Furthermore,
\begin{align}
\label{scinv}\xi \log^* \cE(a) & = \xi \log^*\partial_{a}^{\gamma} \tr(e^{-\beta H_\cC})|_{\gamma=0}\\
\label{appendixcheck} &=\xi \partial_{a}^{\gamma} \log \tr(e^{-\beta
H_\cC})|_{\gamma=0}\\ 
\nonumber &=-\beta \xi \partial_{a}^{\gamma} F |_{\gamma=0},
\end{align}
where (\ref{scinv}) follows from scale-invariance,
(\ref{scale_invariance}), and (\ref{appendixcheck}) follows directly
from the definition of $\log^*$ using (\ref{lg}); see the Appendix for
details.
\end{proof}

Note that in the proof of Theorem \ref{sim1}, $e^{-i\tau H}$
operates on the left of $P_i$ and its conjugate operates on the right,
which leads to the real part, $Re\{\xi \log^* \cD(a)\}$, appearing in
(\ref{simevol}). In the above theorem $e^{-\beta H}$ appears
without its conjugate, so we get the whole of $\xi \log^* \cE(a)$ in
(\ref{thermalresult}).




\section{\label{oldtheorem} New proof of the original theorem}
Finally, we shall show how our moment algebra methods can be used to
prove the original theorem in \cite{GraemeCumulants}. Note that we here prove a slight generalization of  Theorem \ref{old-theorem} in the sense that we allow the cumulants to be taken over arbitrary subsets $a$ of the total collection of pointers $\Omega$.
\begin{proof}[Proof of Theorem \ref{old-theorem}]
\begin{align}
\label{x-def}x(a)&=\frac{\tr \left(P_f \mathcal{U}_{n+1}X_{n}\mathcal{U}_{n}\cdots X_{1}\mathcal{U}_{1}P_i\right)}{|\langle\psi_{f}|U_{n+1}\cdots U_{1}|\psi_{i}\rangle|^{2}},\quad X_{j} = \left\{\begin{matrix} r_{j}e^{\gamma_{j}L_{j}} & \textrm{if} & j\in a\\
\hat{1} & \textrm{if} & j\notin a\end{matrix}\right.\\
\label{y-def}y(a)&= \frac{\tr \left(P_f \mathcal{U}_{n+1}Y_{n}\mathcal{U}_{n}\cdots Y_{1}\mathcal{U}_{1}P_i\right)}{|\langle\psi_{f}|U_{n+1}\cdots U_{1}|\psi_{i}\rangle|^{2}},\quad Y_{j} = \left\{\begin{matrix} e^{\gamma_{j}L_{j}} & \textrm{if} & j\in a\\
\hat{1} & \textrm{if} & j\notin a\end{matrix}\right.\\
\label{z-def}z(a)&=\langle \Pi_{j\in a} r_{j} \rangle=\frac{x(a)}{y(a)},
\end{align}
so the $\bracket$ $z(a)$ is the expectation for pointer measurements
in the subset $a$. Note that the normalizations of the $\bracket$s $x$
and $y$ have been chosen so that $x(\emptyset) = 1$ and
$y(\emptyset) = 1$. For $y$ another convenient property is
that $y(a)|_{\gamma = 0} = 1$ for all $a$. 

The idea of the proof is as follows: if the ratio in (\ref{z-def})
were defined in terms of convolution operations, so we had
$z=x*y^{-1*}$ instead of $z=xy^{-1}$, then this would imply
$\log^{*}z=\log^{*}x-\log^{*}y$, leaving us with the much simpler task of
calculating $\log^{*}x$ and $\log^{*}y$. In fact, it turns out that, by
expanding the $\bracket$ $z$ in powers of the $\gamma$'s, we can achieve this
switch from multiplicative to convolution operations: see
(\ref{star-version}). 

We first prove the
equivalent of (\ref{expansion}) for $z$:
\begin{equation}\label{z-expansion}
 \partial^\gamma_b \log^* z(a)\vert_{\gamma=0} = \left\{ \begin{matrix}
\log^*\tilde{z}(a) & \textrm{if} & b = a,\\
0 & \textrm{if} & \textrm{$b \subset a$, $b\neq a$}\\
\end{matrix}\right.
\end{equation}
where $\tilde{z}(a)=\partial^\gamma_a z(a)\vert_{\gamma=0}$. To this
end, for any $b\subseteq a$ define the $\bracket$ $z_b$ by
$z_b(c) = \partial^{\gamma}_{b\cap c}z(c)|_{\gamma=0}$ for all
$c\subseteq a$. By (\ref{log-explicit}),
\begin{eqnarray}\label{partition-sum}
\partial^\gamma_b\log^* z(a)|_{\gamma=0} & = &\sum_{p\in \pi(a)} (|p|-1)!(-1)^{|p|-1}
 \partial^\gamma_b\prod_{c\in p} z(c)|_{\gamma= 0} \nonumber\\
 & = & \sum_{p\in \pi(a)} (|p|-1)!(-1)^{|p|-1} \prod_{c\in p} \partial^\gamma_{b\cap c}z(c)|_{\gamma=0}=\log^* z_b(a),
\end{eqnarray}
using the fact that $z(c)$ depends only on the $\gamma$'s in
$c$. Since $z_a=\tilde z$, (\ref{z-expansion}) follows for the case
$b=a$. If $b \neq a$, there is some $j \in a$ with $j \notin b$. When
we put $\gamma_j=0$, since (\ref{partition-sum}) shows the term
$e^{\gamma_j L_j}$ is not differentiated, the operator $L_j$ does not
appear in $z_b(a)$. Therefore $r_j$ is not coupled to $\cS$ and $z_b$
factorises on $a$. We conclude that $\log^* z_b(a)=0$.

Equation (\ref{z-expansion}) tells us that $\log^* \tilde z(a)$ is the
first non-vanishing term in the expansion of the cumulant $\log^*
\langle \Pi_{j\in a}r_j \rangle$ in the $\gamma$'s. In other words,
the relevant expansion coefficient of the latter cumulant can be
regarded as the cumulant of the new $\bracket$ $\widetilde{z}$. To
calculate this new cumulant we use the usual law for differentiation
of a product, together with the observation that if $(b_{1},b_{2})$ is
bipartition of $b$ then
\begin{equation*}
[\partial^\gamma_{b_1}x(b)]\big|_{\gamma=0} =
(\Pi_{j\in b_2}\langle r_{j}
\rangle_{\phi_{j}})[\partial^\gamma_{b_1}x(b_1)]\big|_{\gamma=0},\quad [\partial^\gamma_{b_2}\frac{1}{y(b)}]\big|_{\gamma=0} =
[\partial^\gamma_{b_2}\frac{1}{y(b_{2})}]\big|_{\gamma=0},
\end{equation*}
which yields
\begin{align}
\nonumber \tilde z(b)=\partial^\gamma_b \frac{x(b)}{y(b)}\bigg|_{\gamma=0}&=\sum_{b_1 \cup b_2=b} [\partial^\gamma_{b_1}x(b)]\Big[\partial^\gamma_{b_2}\frac{1}{y(b)}\Big]\bigg|_{\gamma=0}\\
\nonumber &=\sum_{b_1 \cup b_2=b}(\Pi_{j\in b_2}\langle r_{j} \rangle_{\phi_{j}})[\partial^\gamma_{b_1}x(b_1)]\Big[\partial^\gamma_{b_2}\frac{1}{y(b_2)}\Big]\bigg|_{\gamma=0}\\
\nonumber &=\sum_{b_1 \cup b_2=b} \tilde x(b_1) \tilde y^{-1*}(b_2)\\
\label{star-version} &={\tilde x}*{\tilde y}^{-1*}(b),
\end{align}
where we have the new $\bracket$s
\begin{align}
\label{x-tilde} \tilde x(b)&=\partial^\gamma_b x(b)\vert_{\gamma=0},\\
\tilde y(b)&= (\Pi_{j\in b} \langle r_{j} \rangle_{\phi_{j}})\partial^\gamma_b
y(b)\vert_{\gamma=0}.
\end{align}
With the aid of (\ref{logprod}) we thus find that the cumulant of the $\bracket$ $\tilde{z}$ can be decomposed as
\begin{align}
\log^* \tilde z(a)=\log^* {\tilde x}(a)-\log^* {\tilde y}(a).
\end{align}

We next turn to the evaluation of $\log^*\tilde{x}(a)$.
Equations (\ref{x-def}) and (\ref{x-tilde}) imply
\begin{equation}
\label{tildex}\tilde{x}(a)=\frac{ \tr \left(P_f \mathcal{U}_{n+1}\widetilde{X}_{n}\mathcal{U}_{n}\cdots \widetilde{X}_{1}\mathcal{U}_{1}P_i\right)}{|\langle\psi_{f}|U_{n+1}\cdots U_{1}|\psi_{i}\rangle|^{2}},\quad \widetilde{X}_{j} = \left\{\begin{matrix} r_{j}L_{j} & \textrm{if} & j\in a\\
\hat{1} & \textrm{if} & j\notin a\end{matrix}\right.
\end{equation}
This should be compared with (\ref{wdef}) in the proof of Theorem
\ref{newtheorem}. Just as in that proof, where we defined
$w_{c_1,c_2}$, here we define $\tilde{x}_{c_1,c_2}$ where $L_j$ is
replaced by $L_j^{\textrm{left}}$ if $j \in c_1$ and by $L_j^{\textrm{right}}$ if $j \in
c_2$. By the same argument, we find $\log^*{\tilde x}(a)=\log^*{\tilde
x}_{a,\emptyset}(a)+\log^*{\tilde x}_{\emptyset,a}(a)$, and direct
calculation gives
\begin{equation}
\tilde{x}_{a,\emptyset}(a)=(-i)^{|a|} (\Pi_{j\in a} \langle r_{j}s_{j} \rangle_{\phi_{j}}) A_w(a),
\end{equation}
and $\tilde{x}_{\emptyset,a}(a)$ gives the complex conjugate.
Thus
\begin{align}
\log^* {\tilde x}(a)=Re\{2(-i)^{|a|} (\Pi_{j\in a} \langle r_{j}s_{j} \rangle_{\phi_{j}}) \log^* A_w(a)\},
\end{align}
and this gives (\ref{main-result}) with the first part of $\xi$
(see (\ref{xi})). The evaluation of the cumulant of
$\log^{*}\tilde{y}(a)$ is analogous to the above, and results in the
second half of $\xi$.
\end{proof}




\section{\label{multisets}Multisets and multiset cumulants}

Our definition of cumulants in Section \ref{convalgebra} was somewhat
unconventional. A more standard definition for the classical cumulant
is
\begin{equation}
\label{standard}
\log^{*}\langle\Pi_{j\in a}X_{j}\rangle = \partial^\gamma_a \log\langle
e^{\sum_{k=1}^{|a|}\gamma_{k}X_{k}}\rangle|_{\gamma=0},
\end{equation}
which is readily seen (Appendix) to be equivalent to ours given
by (\ref{classical}). Thus cumulants can be thought of as
coefficients in a formal power series expansion in the $\gamma$'s. To
a statistician, the expression being expanded is the logarithm of the
moment generating function. To a physicist, a comparison with the
Helmholtz free energy, (\ref{freeenergy}), is compelling, and
indeed there is a strong connection with thermodynamics
\cite{Royer}. A combinatorialist can also stake a claim \cite{Aigner,
Wilf}.

The formal expansion for two variables begins
\begin{equation}
\label{standardexp}
\log\langle e^{\gamma_{1}X_{1}+
\gamma_{2}X_{2}}\rangle=\gamma_1 \langle X_1\rangle+\gamma_2
\langle X_2\rangle-\frac{\gamma_1^2}{2}[\langle X_1^2 \rangle -
\langle X_1 \rangle^2]-\frac{\gamma_2^2}{2}[\langle X_2^2
\rangle - \langle X_2 \rangle^2]-\gamma_1 \gamma_2[\langle X_1X_2
\rangle - \langle X_1 \rangle\langle X_2 \rangle]+ \ldots,
\end{equation}
and the coefficient of $\gamma_1\gamma_2$ is familiar as the
classical cumulant $\log^*\langle X_1X_2 \rangle$. However, this
expansion also forces on one's attention terms in higher powers of the
$\gamma$s, such as $\gamma_1^2(\langle X_1^2 \rangle - \langle X_1
\rangle^2)/2$.  The moment algebra, as we introduced it in
Sec.~\ref{convalgebra}, does not include such terms. However, with a
slight modification of the construction of the algebra, all of these
higher order terms can be incorporated. The natural setting for it is
not subsets of a set, but \emph{multisets} \cite{Blizard}.  In a
multiset, an element may occur any finite number of times. For
instance, if the underlying set is $\Omega_{3} = \{1,2,3\}$, then
$\{1,1\}$, and $\{1,1,1,2,3,3\}$ are two examples of multisets in
$\Omega_{3}$ (Note that, as with ordinary sets, the
ordering does not matter, i.e., $\{1,1,1,2,3,3\} = \{3,1,2,3,1,1\}$ .)

To extend the moment algebra we define $\bracket$s on multisets
rather than subsets, so an $\bracket$ $f$ assigns a complex
number $f(a)$ to every multiset $a$ of $\Omega$. The whole machinery
of formal derivatives and their action on composite functions, as
presented in Sec.~\ref{convalgebra}, goes through essentially
unaltered. In particular, given an $\bracket$ $f$ defined on multisets,
we can define higher order cumulants $\log^{*}f(a)$ for any multiset
$a$. \\

As an example, consider the multiset that consists of $\{1\}$,
$\{1,1\}$, $\{1,1,1\}$, etc.. Suppose $f$ is a function on this
multiset. Then we can apply (\ref{Fstar}) to $a=\{1,1,1\}$, for
instance, so that from
\begin{align*}
\partial_{\{1,1,1\}} \log f =\frac{f^{\prime \prime \prime}}{f}
-3\frac{f^{\prime \prime}f^\prime}{f^2} +2\frac{(f^\prime)^3}{f^3},
\end{align*}
(with primes denoting differentiation by variable 1), we deduce
\begin{align*}
\log^*f(\{1,1,1\})=\frac{f(\{1,1,1\})}{f(\emptyset)}-3\frac{f(\{1,1\})f(\{1\})}{f(\emptyset)^2}+2\frac{f(\{1\})^3}{f(\emptyset)^3}.
\end{align*}
As a special case, we define $f(\{\underbrace{1,1, \ldots
,1}_k\})=\langle X^k\rangle$, for some random variable $X$. Then we
can rewrite the above expression as
\begin{eqnarray}
\log^{*}f(\{1,1,1\}) & = &\langle X^3 \rangle -3\langle X^2 \rangle\langle X
\rangle+2\langle X \rangle^3 \equiv \kappa_3(X).\nonumber
\end{eqnarray}
In this way we obtain the classical cumulant $\kappa_3(X)$, some
others in the series being $\kappa_{1}(X)=\langle X\rangle$ and
$\kappa_2(X)=\langle X^2 \rangle -\langle X \rangle^2$. These
cumulants are widely used in statistics. They have nice properties;
for instance, $\kappa_j(X)=0$ for $j\geq 3$ is a necessary and
sufficient condition for $X$ to be Gaussian \cite{Marcinkiewicz}.

It is easy to see that any multiset cumulant can be obtained from our
original cumulant with distinct variables simply by setting certain
subsets of its variables equal. For instance, $\log^{*}f(\{1,1,1\})$
can be derived from $\log^{*}f(\{1,2,3\})$ by setting every `2' and
`3' to a `1' (and leaving `1's unchanged). Similarly
$\log^{*}f(\{1,1,2\})$ can be obtained from $\log^{*}f(\{1,2,3\})$ by
sending `3' to `1'. 

We can mimic this procedure in the case of multipartite weak
measurements by treating subsets of pointers identically; i.e., we can
couple all the pointers in each subset via the same system observable
and afterwards measure the same pointer observable on each of them. As
a simple case, suppose we have just two pointers and we couple both to
the system through the interaction Hamiltonian $H=s \otimes A$ and
finally carry out the same measurement $r$ on them; note though we
have to give labels to $s$ and $r$ to indicate which pointer they
belong to; so we have $s_1$, $s_2$ and $r_1$, $r_2$, respectively, for
pointers 1 and 2. To avoid the complication of ordering of the
couplings of pointers, let us consider simultaneous weak
measurement. Then Theorem \ref{old-simultaneous-theorem} gives
\begin{align} \label{k2}
\langle r_1r_2 \rangle -\langle r_1\rangle\langle r_2\rangle =
\gamma^2 Re \left\{ \xi \kappa_2(A)_w \right\}
\end{align}
where we have suggestively written $\kappa_2(A)_w$ for the
simultaneous weak value $(A^2)_w-(A_w)^2$ given by
(\ref{simweak}). 

This gives us a procedure for gaining information about the multiset
cumulants of weak values. With thermal weak measurement we can use
this procedure to measure higher order susceptibilities. If we couple
$m_{1}$ independent pointers to observable $A_{1}$, $m_{2}$ pointers
to $A_{2}$, etc, Theorem \ref{sim2} gives a direct relationship
between the correlation of the pointers and the relevant
susceptibility:
\begin{equation}
\log^{*}\langle \Pi_{j=1}^{n}
\Pi_{l_{j}=1}^{m_{j}}r_{j}^{(l_{j})}\rangle = -\beta (\Pi_{j=1}^{n}
\Pi_{l_{j}=1}^{m_{j}}\gamma_{j, l_{j}})\xi \frac{\partial^{m_{1}} }{\partial
\gamma_{1}^{m_{1}}}\cdots\frac{\partial^{m_{n}}}{\partial
\gamma_{n}^{m_{n}}}F|_{\gamma=0}.
\end{equation}




\section{Conclusions}

In physics we often study the effects of weak coupling between
systems. By focussing on the effect of one system upon the other, weak
measurement gives a way of understanding the nature of such
interactions. If we weakly couple two systems, $\cS$ and $\cP$,
say, and then carry out a strong measurement on $\cP$, the effects
of the coupling can be expressed in terms of weak values
\cite{AAV88,AharonovRohrlich05}. 
When $P$ consists of a very simple system, e.g., a ``pointer'' or particle
in a given initial state, the measurement results depend on weak
values in a very simple way (\ref{weakvalue}).
As $\cP$ becomes more complicated, the
dependency becomes rapidly more complicated, and in fact already
assumes a highly baroque form when $\cP$ consists of two pointers
applied at different times (see the Appendix of
\cite{GraemeCumulants}). However, this complication vanishes if one
takes cumulants: the cumulant of the measured variables and the
cumulant of the weak values are once more simply related. The aim of
this paper has been to give proofs of this fact that illuminate why
this phenomenon occurs.

The proofs of our various theorems repeatedly use two properties of
cumulants: first, that they are logarithms in the moment algebra
$\cM_n$, and turn a convolution product into a sum; second, that they
vanish on maps that factorise, i.e., that can be written as a product
of two maps that are defined on disjoint sets of variables. All the
maps that we construct are elements of $\cM_n$, which is
therefore the natural setting for the proofs. This algebra is in
itself an interesting object. Although some of its features
have been thoroughly described in the literature, we are not aware
of any explicit formulation of the algebra as an object in its own
right.  We feel that it deserves this recognition, because of the
simplicity of the definitions of its operations, and the surprising
richness of the structure that this gives rise to.

Cumulants have long played a role in statistical mechanics \cite{Kubo,
KahnUhlenbeck, Sylvester, Royer}, where they are used to simplify
perturbation expansions. This suggests alternative weak measurement
scenarios. In this spirit we consider a collection of pointers that
reach thermal equilibrium with a probed system, and call this a
``thermal weak measurement". By its very nature, this excludes pre-
and postselection, which are standard components of weak
measurement. We lose thereby some of the strengths of weak
measurement: many of the more intriguing phenomena in the standard
setting arise from postselection. However, we retain the advantages of
minimal perturbation of a system, and the possibility of applying
several probes simultaneously or sequentially opens up some new
territory for exploration.




\section{Appendix: Defining cumulants via generating functions} 

Here we show that the standard definition of the classical cumulant
via the generating function (\ref{standard}) is equivalent to
our definition (\ref{lg}) of $\log^*f$ for the $\bracket$
$f(a)=\langle \Pi_{j \in a} X_j \rangle$.  In other words, we wish to
show that, for any multiset $a$,
\begin{equation}
\label{avavn}
\log^{*}f(a) = \partial_{a}^{\gamma}\log h(\gamma)|_{\gamma=0},
\end{equation}
where
\begin{equation}
\label{fforv}
h(\gamma)=\langle e^{\sum_k \gamma_k X_k} \rangle.
\end{equation}
We first note that 
\begin{equation}
\label{sfdbn}
\partial_{a}^{\gamma}\log h = \Lambda(h,\partial_{1}^{\gamma}h,\partial_{2}^{\gamma}h,\partial^{\gamma}_{1,2},\ldots),
\end{equation}
where $\Lambda$ is a function of all the relevant partial derivatives
of $h$, obtained via the chain rule when we apply
$\partial_{a}^{\gamma}$ to the logarithm. The same function $\Lambda$
appears when we express the cumulant $\log^{*}f(a)$ in terms of
formal derivatives
\begin{equation}
\label{ankldv}
\begin{split}
\log^{*}f(a) = &  \partial_{a}\log f(a)\\
 = & \Lambda(f,\partial_{1}f,\partial_{2}f,\partial_{1,2}f,\ldots)\\
 = &  \Lambda(f(\emptyset),f(1),f(2),f(1,2),\ldots).
\end{split}
\end{equation}
To obtain (\ref{avavn}) from (\ref{sfdbn}) and (\ref{ankldv}) we only
need observe that $\partial_{c}^{\gamma}h|_{\gamma =0} = \langle
\Pi_{j \in c} X_j \rangle = f(c)$.

This equivalence of definitions can be extended to other functions $h$
if we take $f(a) = \partial_{a}^{\gamma}h|_{\gamma = 0}$. For example,
let $h(\gamma)=\tr ( e^{-\beta H_\cC} )=\tr ( e^{-\beta H_\cS -
\sum_{j}\gamma_{j}A_{j})} )$. Then the above arguments prove
(\ref{appendixcheck}), which occurs in the proof of Theorem
\ref{sim2}.  Note that the arguments extend without difficulty to
non-commutative observables.

\end{document}